%% file: main.tex
\newif\ifLIPICS
\LIPICSfalse
\ifLIPICS
    \documentclass[a4paper,UKenglish,cleveref, autoref, thm-restate,numberwithinsect]{lipics-v2021}
    \usepackage{paper}
\else
    \documentclass[11pt]{article}
    \usepackage{paper}
    \usepackage{cleveref, thm-restate}
\fi

\newif\ifDRAFT 
\DRAFTfalse
\usepackage{ulem}

\ifDRAFT
    \usepackage{lineno}
    \linenumbers
\fi

\title{
Lower Bounds for Pseudo-Deterministic Counting in a Stream
}
\ifLIPICS

    \author{Vladimir Braverman}{Rice University}{vb21@rice.edu}{}{Work partially supported by ONR Award N00014-18-1-2364 and NSF awards 1652257, 1813487 and 2107239.} 

    \author{Robert Krauthgamer}{Weizmann Institute of Science}{robert.krauthgamer@weizmann.ac.il}{https://orcid.org/0009-0003-8154-3735}
    {Work partially supported by ONR Award N00014-18-1-2364,
      by a Weizmann-UK Making Connections Grant,
      by a Minerva Foundation grant,
      and the Weizmann Data Science Research Center.} 

    \author{Aditya Krishnan}{Pinecone}{aditya@pinecone.io}{}{Work partially done while the author was at Johns Hopkins University and supported by the MINDS Data Science Fellowship.} 

    \author{Shay Sapir}{Weizmann Institute of Science \and \url{https://sites.google.com/view/shaysapir}}{shay.sapir@weizmann.ac.il}{https://orcid.org/0000-0001-7531-685X}{This research was partially supported by the Israeli Council for Higher Education (CHE) via the Weizmann Data Science Research Center.}

    \authorrunning{V. Braverman, R. Krauthgamer, A. Krishnan and S. Sapir} %TODO mandatory. First: Use abbreviated first/middle names. Second (only in severe cases): Use first author plus 'et al.'
    
    \Copyright{Vladimir Braverman, Robert Krauthgamer, Aditya Krishnan and Shay Sapir} %TODO mandatory, please use full first names. LIPIcs license is "CC-BY";  http://creativecommons.org/licenses/by/3.0/
    
    \input{ccsdesc_PD_streams} % in seperated tex file

    
    \category{Track A: Algorithms, Complexity and Games}
    
    \relatedversion{} %optional, e.g. full version hosted on arXiv, HAL, or other respository/website
    \relatedversiondetails{}{https://arxiv.org/abs/2303.16287}
    \nolinenumbers %uncomment to disable line numbering

\EventEditors{Kousha Etessami, Uriel Feige, and Gabriele Puppis}
\EventNoEds{3}
\EventLongTitle{50th International Colloquium on Automata, Languages, and Programming (ICALP 2023)}
\EventShortTitle{ICALP 2023}
\EventAcronym{ICALP}
\EventYear{2023}
\EventDate{July 10--14, 2023}
\EventLocation{Paderborn, Germany}
\EventLogo{}
\SeriesVolume{261}
\ArticleNo{69}
\else
    \author{
    Vladimir Braverman%
    \thanks{Work partially supported by ONR Award N00014-18-1-2364 and NSF awards 1652257, 1813487 and 2107239}\\ Rice University \\ \texttt{vb21@rice.edu} \and 
    Robert Krauthgamer%
    \thanks{Work partially supported by ONR Award N00014-18-1-2364,
      by a Weizmann-UK Making Connections Grant,
      by a Minerva Foundation grant,
      and the Weizmann Data Science Research Center.
      } \\ Weizmann Institute of Science \\ \texttt{robert.krauthgamer@weizmann.ac.il} \and 
    Aditya Krishnan%
    \thanks{Work partially done while the author was at Johns Hopkins University and supported by the MINDS Data Science Fellowship.} \\Pinecone \\ \texttt{aditya@pinecone.io} \and 
    Shay Sapir%
    \thanks{This research was partially supported by the Israeli Council for Higher Education (CHE) via the Weizmann Data Science Research Center.} \\ Weizmann Institute of Science \\ \texttt{shay.sapir@weizmann.ac.il} }

    \date{}
\fi

\begin{document}
\maketitle

\normalem

\begin{abstract}
Many streaming algorithms provide only a high-probability relative approximation. These two relaxations, of allowing approximation and randomization, seem necessary -- for many streaming problems, both relaxations must be employed simultaneously, to avoid an exponentially larger (and often trivial) space complexity. A common drawback of these randomized approximate algorithms is that independent executions on the same input have different outputs, that depend on their random coins. \emph{Pseudo-deterministic} algorithms combat this issue, and for every input, they output with high probability the same ``canonical'' solution.

We consider perhaps the most basic problem in data streams, of counting the number of items in a stream of length at most $n$. Morris's counter [CACM, 1978] is a randomized approximation algorithm for this problem that uses $O(\log\log n)$ bits of space, for every fixed approximation factor (greater than $1$). Goldwasser, Grossman, Mohanty and Woodruff [ITCS 2020] asked whether pseudo-deterministic approximation algorithms can match this space complexity. Our main result answers their question negatively, and shows that such algorithms must use $\Omega(\sqrt{\log n / \log\log n})$ bits of space.

Our approach is based on a problem that we call \emph{Shift Finding}, and may be of independent interest. In this problem, one has query access to a shifted version of a known string $F\in\{0,1\}^{3n}$, which is guaranteed to start with $n$ zeros and end with $n$ ones, and the goal is to find the unknown shift using a small number of queries. We provide for this problem an algorithm that uses $O(\sqrt{n})$ queries. It remains open whether $\poly(\log n)$ queries suffice; if true, then our techniques immediately imply a nearly-tight $\Omega(\log n/\log\log n)$ space bound for pseudo-deterministic approximate counting.

\end{abstract}

\input{Introduction}

\input{Preliminaries}

\input{Connection_ShiftFind_and_PDcounting}
\input{LBforPD_ApproximateCounting}

% Shift finding
\input{Shift_Finding_Alg}

\ifLIPICS
    \bibliographystyle{plainurl}% the mandatory bibstyle
\else
    \bibliographystyle{alphaurl}
\fi
\bibliography{references.bib}
\end{document}

%% file: ccsdesc_PD_streams.tex
\ccsdesc[500]{Theory of computation~Streaming, sublinear and near linear time algorithms}
\ccsdesc[300]{Theory of computation~Lower bounds and information complexity}
\ccsdesc[100]{Theory of computation~Pseudorandomness and derandomization}

%% file: Introduction.tex
\section{Introduction}
\label{sec:intro} 

Computing over data streams is a rich algorithmic area
that has developed enormously, 
and actually started with the simple-looking problem 
of approximate counting~\cite{DBLP:journals/cacm/Morris78a}. 
Let us first recall the streaming model:
The input is a stream, i.e., a sequence of items, 
and the goal is to compute a pre-defined function of these items, 
such as the number of items (or number of the distinct items),
while making one sequential pass over the stream (or sometimes a few passes). 
Many useful functions actually depend on the items as a multiset,
i.e., ignoring their order, or even only on their frequencies 
(like the famous $\ell_p$-norm of the frequency vector).
Another possible goal is to produce a sample, rather than computing a function,
e.g., to produce a uniformly random item.

The primary measure of efficiency for streaming algorithms
is their space complexity,
and for many problems, researchers have designed space-efficient algorithms, 
often with space complexity that is even polylogarithmic in the input size.
However, this comes at a price ---
these algorithms are usually randomized (and not deterministic)
and/or compute an approximate solution (rather than exact one).
In fact, oftentimes both relaxations are needed in order to achieve low space complexity.  
For example, to count the number of items in a stream of length at most $n$,
there is a
randomized approximation algorithm 
using $O(\loglog n)$ bits of space,
but algorithms that are exact or deterministic must use $\Omega(\log n)$ bits~\cite{DBLP:journals/cacm/Morris78a}. 
Another example is the $\ell_2$-norm of the frequency vector of items from a ground set $[d]$
(or equivalently, of a $d$-dimensional vector under a sequence of additive updates) 
---
there is a randomized approximation algorithm that uses $O(\log d)$ bits of space,
but algorithms that are exact or deterministic must use $\Omega(d)$ bits of space~\cite{DBLP:conf/stoc/AlonMS96}.

Gat and Goldwasser~\cite{DBLP:journals/eccc/GatG11} initiated
the study of \emph{pseudo-deterministic} algorithms,
which informally means that when run (again) on the same input,
with high probability they produce exactly the same output.
This notion combats a potential issue with randomized algorithms,
that independent executions on the same input might return different outputs,
depending on the algorithm's coin tosses.
Many known streaming algorithms suffer from this issue,
which is a serious concern for some users and applications. 
Pseudo-deterministic algorithms were later considered in the streaming model
by Goldwasser, Grossman, Mohanty and Woodruff~\cite{DBLP:conf/innovations/GoldwasserGMW20},
and these are formally defined as follows.

\begin{definition}\label{def:PD_streaming}
A streaming algorithm $A$ is \emph{pseudo-deterministic} (PD)
if there is a function $F(\cdot)$ defined on inputs of $A$ (streams),
such that for every stream $\sigma$,
$$
  \P [A(\sigma)=F(\sigma)] \geq 9/10,
$$
where the probability is over the random choices of the algorithm.
We shall refer to $F$ as the \emph{canonical function} of algorithm $A$.%
\footnote{
The canonical function $F$ depends on the order arrival of the stream items.
In an alternative definition, the canonical function depends on the items only as a multiset, i.e., ignoring their order in the stream. % $I_\sigma$.
These two definitions are equivalent in the setting of approximate counting, which is the focus of our work.
}
\end{definition}

We focus on
\emph{estimation problems},
which ask to approximate a numerical value, 
and are very popular in the streaming model.
For such problems, the notion of PD relaxes the exact setting and the deterministic one,
since exact algorithms have one canonical output (the exact numerical value), and hence they are PD.
Thus the known lower bounds for these settings do not apply for PD algorithms,
and a central question, identified in~\cite{DBLP:conf/innovations/GoldwasserGMW20}, remains open:

\begin{center}
  \emph{Are there efficient PD streaming algorithms for estimation problems?}
\end{center}

Currently, no lower bounds are known for natural
estimation problems,
although for several search problems,
like reporting an element from a stream with deletions
(equivalently, an index from the support of the frequency vector), 
it is known that lower bounds for deterministic algorithms
extend to PD algorithms~\cite{DBLP:conf/innovations/GoldwasserGMW20}.

\subsection{Main Result: Approximate Counting}

Perhaps the most basic problem in the streaming model is to count the number of stream items.
Exact counting, i.e., computing the number of items exactly, 
requires $\Theta(\log n)$ bits of space when the stream has length at most $n$,
even for randomized algorithms with some error probability.
Work by Morris~\cite{DBLP:journals/cacm/Morris78a}, later refined in~\cite{DBLP:journals/bit/Flajolet85,DBLP:journals/mst/GronemeierS09,NelsonYu22},
showed that the number of stream items
can be $(1+\epsilon)$-approximated with probability $9/10$
using $O_\epsilon(\loglog n)$ bits of space,
where $\epsilon>0$ is arbitrary but fixed.
Throughout, we refer to multiplicative approximation,
and use the notations $O_c(\cdot)$ and $\Omega_c(\cdot)$
to hide factors that are polynomial in $c$. 
Morris's algorithm has found many applications,
both in theory and in practice~\cite{DBLP:journals/corr/abs-1805-00612_Lumbroso,NelsonYu22}.
An open question stated explicitly by Goldwasser, Grossman, Mohanty and Woodruff~\cite{DBLP:conf/innovations/GoldwasserGMW20} is 
whether there is a PD algorithm for this problem using $O(\loglog n)$ bits of space.
We answer their question negatively, by proving the following lower bound. 

\begin{theorem}[Main Result]\label{cor:LB_PD_counting}%maybe bad name for the label
For every $c,n>1$,
every PD streaming algorithm that $c$-approximates the number of items in a stream of length at most $(c+1)n$
must use $\Omega_c(\sqrt{{\log n}/{\loglog n}})$ bits of space.
\end{theorem}

To be more precise, our lower bound is actually 
$\Omega\big(\tfrac{\log n}{\sqrt{\log n\loglog (cn)} + \log c}\big)$,
which is still $\Omega(\sqrt{\tfrac{\log n}{\loglog n}})$
as long as $c<2^{\sqrt{\log n \loglog n}}$. 
Previously, there was a large gap for this problem,
between $O(\log n)$ bits (by a deterministic algorithm)
and $\Omega(\loglog n)$ bits (from the randomized setting)~\cite{NelsonYu22}.
See \Cref{table:Approx_Count_Bounds} for a summary of the known bounds.

\begin{table*}[t] % or [h]?
\caption{
\label{table:Approx_Count_Bounds} 
Known space bounds (in bits) for $2$-approximate counting in a stream of length at most $n$.
Folklore bounds are stated without a reference. 
}
\begin{center}
\begin{tabulary}{\textwidth}{|L  | rl | rl|}
\hline
Algorithms
& \multicolumn{2}{c|}{Upper bound}
&  \multicolumn{2}{c|}{Lower bound}
\\ 
\hline
Exact or deterministic
& $O(\log n)$ &
& $\Omega(\log n)$ & \\ 
Randomized and approximate
& $O(\loglog n)$ & \cite{DBLP:journals/cacm/Morris78a}
& $\Omega(\loglog n)$ & \cite{NelsonYu22}
\\
Pseudo-deterministic
& $O(\log n)$  & 
& $\Omega(\sqrt{{\log n}/{\loglog n}})$ & [Thm.~\ref{cor:LB_PD_counting}] \\
Pseudo-deterministic &               &
& $\Omega(\log n)$ & \cite{grossman2023tight} \\
\hline
\end{tabulary}
\end{center}
\end{table*}

Our proof analyzes the promise variant 
of $c$-approximate counting for streams of length at most $(c+1)n$,
which we denote by $\prac$;
this variant asks to distinguish whether the number of stream items
is $\leq n$ or $>cn$ (see Definition \hyperref[sec:preliminaries]{2.1}). % BUG: \Cref{def:stream_approx_count}) lead to the wrong place
A crucial property of PD algorithms is that they
have to be PD also for inputs in the range $[n+1,cn]$ (i.e., outside the promise).
We rely on this property of PD algorithms to prove the following result,
which immediately yields \Cref{cor:LB_PD_counting} as a corollary.

\begin{theorem}[Main Result]\label{thm:LB_PD_counting}
For every $c,n>1$, 
every PD streaming algorithm for problem $\prac$ must use $\Omega_c(\sqrt{ {\log n}/{\loglog n}})$ bits of space.
\end{theorem}

Our proof of Theorem~\ref{thm:LB_PD_counting}
appears in Section~\ref{sec:LB_PD_counting}.
It is based on a problem that we call Shift Finding,
which may be of independent interest, 
as it is very natural and likely to find connections to other problems.
In addition, it can potentially lead to a near-tight
$\Omega(\log n/\loglog n)$ lower bound for PD streaming,
by simply improving our algorithmic result for Shift Finding. 
A very recent independent work by
Grossman, Gupta and Sellke~\cite{grossman2023tight} 
shows a tight $\Omega(\log n)$ bound for $\prac$,
using a very different technique,
which views the PD streaming algorithm as a Markov chain with a limited number of states.

\subsection{Main Technique: The Shift Finding Problem}

Our main result relies on \emph{algorithms} for the shift Finding problem $\prsf$, which is defined below. 
Let us first introduce some basic terminology.
A function $F:[m]\to\{0,1\}$ can also be viewed as a string $F\in\{0,1\}^m$, 
and vice versa, and we sometimes use these interchangeably. 
Given $s\in[0,n]$, let the shifted version of this $F$
be the function $F_{s}: x\mapsto F(s+x)$, with a properly restricted domain,
see \Cref{sec:preliminaries}.

\begin{definition}[Shift Finding] \label{def:shift_finding}
Let $c,n>1$.
In problem $\prsf$,
the input is a string $P\in \{0,1\}^{(c-1)n}$,
and one has query access to a string $F_{s^*}$ that is the concatenation
of $n-s^*$ zeros, then $P$, and finally $s^*$ ones, 
for an unknown $s^*\in [0,n]$. 
Thus, a query for $x\in[0,cn]$ returns $F_{s^*}(x)$.
The goal is to output $s^*$. 
\end{definition}

The measure of complexity of an algorithm for this problem
is the number of queries that it makes to $F_{s^*}$.
A randomized algorithm is required to be correct (in its output $s^*$)
with probability $9/10$.

This problem may be also of independent interest.
In a different variant of shift finding,
the input is a random string $c\in\set{0,1}^n$ and a vector $x$ that is obtained from the string $c$ by a cyclic shift $\tau$ and some noise (random bit flips),
and the goal is to compute the shift $\tau$ with high probability.
This problem is related to GPS synchronization, see~\cite{DBLP:conf/mobicom/HassaniehAKI12,DBLP:conf/soda/AndoniIKH13} for more details.
There is a sublinear time algorithm for this problem, running in time roughly $O(n^{0.641})$~\cite{DBLP:conf/soda/AndoniIKH13}.
One main difference is that in our \Cref{def:shift_finding},
one string is completely known to the algorithm, and the only concern is the number of queries to the second string.

\subsubsection{Connection to PD Counting}

We show that an algorithm for Shift Finding ($\prsf$) implies
a space lower bound for PD streaming algorithm for counting ($\prac$).

\begin{theorem}\label{thm:connection_Shift_Counting}
Let $c,n>1$, and suppose that the Shift Finding problem $\prsf$
admits a randomized algorithm that makes at most $q=q(c,n)$ queries (possibly adaptive).
Then, every PD streaming algorithm for the approximate counting problem $\prac$
must use $\Omega(\tfrac{\log n}{\log q})$ bits of space.
\end{theorem}

It immediately follows that if the Shift Finding problem $\prsf$ can be solved using $\polylog(n)$ queries (for fixed $c>1$), then PD approximate counting requires $\Omega(\tfrac{\log n}{\loglog n})$ bits of space.
However, our current upper bound for Shift Finding is $q=O(\sqrt{cn})$ queries (\Cref{thm:shift_finding_alg}) and is not strong enough to yield a nontrivial lower bound for PD approximate counting. 

Therefore, to prove our main lower bound (\Cref{thm:LB_PD_counting}),
we revert to a generalization of \Cref{thm:connection_Shift_Counting}
where the Shift Finding algorithm is still given an instance of problem $\prsf$
(namely, a string $F$ and query access to $F_{s^*}$), 
but reports a small set $R\subset[0,n]$ (say of size $|R|\leq t$) 
that contains the unknown shift (i.e., $s^*\in R$).
This algorithm may be randomized provided that it is PD,
and its canonical function maps each instance of problem $\prsf$
to a set $R$ of size $t$ that contains $s^*$.

\begin{theorem}\label{thm:generalized_connection_Shift_Counting}
Let $c,n>1$,
and suppose there is a PD algorithm $Q$ that,
given an instance of problem $\prsf$,
makes at most $q=q(c,n)$ queries (possibly adaptive) to $F_{s^*}$
and its canonical function $M$ maps the input to a set
$R\subset [0,n]$ of size $t=t_c(n)$ that contains $s^*$.
Then every PD streaming algorithm for problem $\prac$ must use
$\Omega(\tfrac{\log (n/t)}{\log q})$ bits of space.
\end{theorem}

We use \Cref{thm:generalized_connection_Shift_Counting},
(more precisely its proof arguments rather than its statement) 
to prove our main result (\Cref{thm:LB_PD_counting}),
see \Cref{sec:LB_PD_counting}.
At a high level,
the proof of \Cref{thm:LB_PD_counting} proceeds by splitting into two cases,
depending on the canonical function $F$.
Roughly speaking, in one case we show a Shift Finding algorithm
that returns a set of size $t=n/2^{\sqrt{\log n}}$
% the relaxed version of Shift Finding is solved 
using $q=O(\log n)$ queries by binary search,
and in the other case an algorithm to find the shift (i.e., $t=1$)
with probability $9/10$ using $q=2^{\sqrt{\log n}}$ uniformly random queries.

As a corollary of \Cref{thm:connection_Shift_Counting},
we get that the \emph{tracking} version of approximate counting
must use $\Omega(\log n)$ bits of space,
which is tight with a straightforward deterministic counting.
Tracking means that the algorithm produces an output
after every stream item rather than at the end of the stream,
and with probability $9/10$, all the outputs are simultaneously correct (i.e., approximate the number of items seen so far).

\begin{corollary}[Tracking]\label{thm:LB_PD_tracking_counting}
For every $c,n>1$, every PD tracking algorithm that $c$-approximates the number of items in a stream of length $(c+1)n$ must use $\Omega(\log n)$ bits of space.
\end{corollary}

In contrast, for standard randomized algorithms, 
there is a tracking algorithm for  $(1+\epsilon)$-approximate counting that uses $O_\epsilon(\loglog n)$ bits of space, for any fixed $\epsilon>0$ \cite{NelsonYu22}.
\Cref{thm:LB_PD_tracking_counting} follows by an easy modification of the proof of \Cref{thm:connection_Shift_Counting}.
That proof uses $O(\log q)$ repetitions of a PD streaming algorithm,
and then employs a union bound on $q$ input streams,
which is not necessary for tracking algorithms and thus the bound follows.

A more direct argument is essentially by equivalence to exact counting. 
For a stream with $s<n$ items, 
the state of a PD tracking algorithm with canonical function $F$
can be used to compute $s$, as follows.
Simulate insertion of more items to the stream until the output of the algorithm changes to $1$ (which corresponds to the first $1$ in $F_s$), from which $s$ can be computed.

\subsubsection{An Algorithm for Shift Finding}

Consider a special case of the Shift Finding problem $\prsf$,
where the input string $P$
is a run of zeros followed by a run of ones (viewed as a function, it is a step function);
then the algorithm can perform a binary search using $O(\log (cn))$ queries,
and find the unique location where $F_{s^*}$ switches from value $0$ to $1$,
and hence recover $s^*$. 
At the other extreme, suppose the input string $P$ is random;
then with high probability
every set of $O(\log n)$ queries from $P$ (and thus from $F_{s^*}$)
will be answered differently (viewed as a string in $\{0,1\}^{O(\log n)}$). 
Based on these observations, one may hope that problem $\prsf$
admits an algorithm that makes $\polylog (cn)$ queries.
We leave this as an open question
and prove a weaker bound of $O(\sqrt{cn})$ queries.

\begin{theorem}[Shift Finding Algorithm]\label{thm:shift_finding_alg}
There is a deterministic algorithm for problem $\prsf$ that makes $O(\sqrt{cn})$ queries.
\end{theorem}

A key observation in our result, that may be useful in future work,
is that for every shift $s^*$
there is a ``short witness'' that uses exactly $2$ queries.
We formalize this as verifying a given guess $s$ for the shift $s^*$. 

\begin{lemma}[Short Witness]\label{lem:shift_verifier_witness_2queries}
There is a deterministic algorithm that,
given as input an instance of problem $\prsf$ and $s<n$,
makes $2$ queries to $F_{s^*}$ and returns 'yes' if $s=s^*$ and 'no' otherwise.
\end{lemma}

The proofs of \Cref{thm:shift_finding_alg} and \Cref{lem:shift_verifier_witness_2queries} appear in \Cref{sec:shift_finding_alg}.
At a high level, 
the Shift Finding algorithm in \Cref{thm:shift_finding_alg} queries the set $\{F_{s^*}(0),F_{s^*}(\sqrt{cn}),F_{s^*}(2\sqrt{cn}),...,F_{s^*}(cn)\}$, and then uses the short witness (\Cref{lem:shift_verifier_witness_2queries}) to check every feasible $s\in [n]$ (i.e., that agrees with the query answers).
Following an observation by Peter Kiss, we are able to improve our Shift Finding algorithm to use only $O((cn)^{1/3}\log n)$ queries; details omitted.

\subsection{Related Work}
\label{sec:related}

\paragraph*{Pseudo-deterministic algorithms}
The notion of pseudo-deterministic algorithms was introduced by~\cite{DBLP:journals/eccc/GatG11} (they originally called them Bellagio algorithms), followed by a long sequence of works that studied it in different models~\cite{DBLP:conf/innovations/GoldreichGR13,DBLP:journals/eccc/Grossman15,DBLP:journals/eccc/GoldwasserG15,DBLP:conf/stoc/OliveiraS17,DBLP:journals/corr/Holden17,DBLP:conf/innovations/GoldwasserGH18,DBLP:conf/mfcs/DixonPV18,DBLP:conf/approx/OliveiraS18,DBLP:journals/corr/abs-1910-00994_GoemansGH2019,DBLP:conf/soda/GrossmanL19,DBLP:journals/eccc/Goldreich19,DBLP:conf/innovations/GoldwasserGMW20,DBLP:conf/stoc/LuOS21,DBLP:conf/innovations/00020V21,DBLP:conf/coco/GoldwasserIPS21,DBLP:conf/approx/GhoshG21,DBLP:conf/stoc/0002PWV22}.
In the streaming and sketching models, \cite{DBLP:conf/innovations/GoldwasserGMW20} proved strong lower bounds for finding a non-zero entry in a vector (given in a stream with deletions), and for sketching $\ell_2$-norms.
Another related setting is that of
sublinear time computation. 
Under certain assumptions, PD algorithms (in the sublinear time region) were shown to admit the following relation with deterministic algorithms -- if for a certain problem there is a PD algorithm using $q$ queries, then there is a deterministic algorithm using $O(q^4)$ queries~\cite{DBLP:conf/innovations/GoldreichGR13}.
The techniques of~\cite{DBLP:conf/innovations/GoldreichGR13} do not seem to extend to streaming algorithms.

\paragraph*{Adaptive adversarial streams}
In this setting, the stream items are chosen adversarially and depend on past outputs of the streaming algorithm (i.e., the stream is adaptive)~\cite{DBLP:journals/jacm/Ben-EliezerJWY22}.
This model is considered to be between PD algorithms and the standard randomized setting, in the sense that for streams of length $m$, amplifying a PD algorithm to success probability $1-\tfrac{1}{10m}$ (by $O(\log m)$ repetitions and taking the median) guarantees (by a union bound) that the algorithm outputs the canonical solution after every stream item with probability $9/10$, thus the adversary acts as an oblivious one (the adversary knows in advance the output of the streaming algorithm, which is the canonical function).
For approximate counting, adaptive streams and standard (oblivious) streams are equivalent (since the stream items are identical) and thus admit an algorithm using $O(\loglog n)$ bits of space.

There is a vast body of work designing algorithms for adaptive streams, 
but not much is known in terms of lower bounds.
Lower bounds are known for some search problems, like finding a spanning forest in a graph undergoing edge insertions and deletions, but also for graph coloring~\cite{DBLP:conf/innovations/ChakrabartiGS22}.
Regarding estimation problems, the only lower bound we are aware of is
for some artificial problem~\cite{DBLP:conf/crypto/KaplanMNS21}.  
Recently, Stoeckl~\cite{doi:10.1137/1.9781611977554.ch32_Stoeckl} showed a lower bound on streaming algorithms that use a bounded amount of randomness, conditioned on a lower bound for PD algorithms.
In the related model of linear sketching, Hardt and Woodruff~\cite{10.1145/2488608.2488624_HW13} 
showed lower bounds on the dimensions of sketching algorithms, which
applies to many classical problems, like $\ell_p$-norm estimation and heavy hitters.

%%% Local Variables:
%%% mode: latex
%%% TeX-master: "main"
%%% End:

%% file: Preliminaries.tex
\section{Preliminaries}
\label{sec:preliminaries}

\begin{definition}[Approximate counting]\label{def:stream_approx_count}
Let $c,n>1$. In problem $\prac$,
the input is a stream of $l\leq (c+1)n$ identical items.
The goal is to output $0$ if $l\leq n$ and $1$ if $l>cn$ 
(and otherwise the output can be either $0$ or $1$). 
\end{definition}

Let $A$ be a PD algorithm for problem $\prac$,
and let $F:[0,(c+1)n]\to\set{0,1}$ be the canonical function of $A$.
Thus, there is a fixed string $P\in\{0,1\}^{(c-1)n}$ such that
\[
F(x) = 
\begin{cases}
0 & \text{if } x\in[0,n]; \\
1 & \text{if } x\in [cn+1,(c+1)n]; \\
P(x-n) & \text{otherwise.}
\end{cases}
\]
For $s^*\in[0,n]$, let $F_{s^*}:[0,(c+1)n-s^*]\to\set{0,1}$ be a shifted version of $F$,
namely the function $F_{s^*}: x\mapsto F(s^*+x)$.
We use these notations throughout the paper.

Our proofs are based on a reduction
from a simple one-way communication problem, called MESSAGE and denoted $\prmsg$,
where Alice's input $x$ is from an alphabet $\Sigma$ that is fixed in advance,
Bob has no input,
and the goal is that Bob outputs $x$ with probability at least $2/3$.
It is well known that this problem requires $\Omega(\log |\Sigma|)$ bits of communication,
even for randomized protocols using shared randomness.
We provide a proof for completeness.

\begin{lemma}\label{lem:comm_problem_message}
    For every alphabet $\Sigma$, every one-way communication protocol (even with shared randomness) for problem $\prmsg$ must use $\Omega(\log |\Sigma|)$ bits of communication.
\end{lemma}
\begin{proof}
    Let $\mathcal{A}$ be a protocol for problem $\prmsg$.
    For a random string $r$ representing the randomness of $\mathcal{A}$, let $\Sigma_r\subset \Sigma$ be the set of all $s\in \Sigma$ for which Bob correctly recovers $s$.
    Let $r^*$ be a string maximizing $|\Sigma_r|$,
    then by averaging, $|\Sigma_{r^*}|\geq \tfrac{2}{3}|\Sigma|$.
    Consider an instance of $\mathcal{A}$ that uses $r^*$ as its random string.
    Assume by contradiction that the number of communication bits is less than $\log |\Sigma_{r^*}|$,
    then by the pigeonhole principle there are two distinct inputs $s, s'\in \Sigma_{r^*}$
    such that $\mathcal{A}(s)$ and $\mathcal{A}(s')$ result in the same message.
    Bob then
    cannot distinguish between (i.e., has the same output distribution for) $s$ and $s'$, a contradiction.
    Hence, the number of bits of communication is at least $\log |\Sigma_{r^*}| = \Omega(\log |\Sigma|)$.
\end{proof}

%%% Local Variables:
%%% mode: latex
%%% TeX-master: "main"
%%% End:

%% file: Connection_ShiftFind_and_PDcounting.tex
\section{Lower Bounds for PD Approximate Counting via Shift Finding}\label{sec:reductionish_Shift_counting}

In this section, we prove \Cref{thm:connection_Shift_Counting}.
The proof involves three problems from different settings:
(a) PD approximate counting in the streaming model;
(b) Shift Finding in the query-access model; and
(c) MESSAGE in one-way communication with shared randomness.
The proof essentially shows that if there is an algorithm for Shift Finding that makes only $q$ queries and also a streaming algorithm for PD approximate counting that uses $b$ bits of space,
then MESSAGE can be solved using $O(b\log q)$ bits of communication.
Combining this bound with the well-known lower bound for MESSAGE in \Cref{lem:comm_problem_message} yields a lower bound for $b$.

A core idea in the proof is that an execution of a PD streaming algorithm $A$ for the approximate counting problem $\prac$
on a stream with $s^*$ insertions, can be used (even without knowing $s^*$,
by making additional insertions and then querying the streaming algorithm $A$) 
to provide query access to the shifted function $F_{s^*}: x\mapsto F(s^*+x)$.
This query access, along with a query-efficient algorithm for the Shift Finding problem $\prsf$, is then used to solve an instance of the MESSAGE problem $\prmsg$.

In fact, we prove the following theorem,
which holds for each string $F$ separately
(rather than a bound that depends on the worst-case $F$),
and yields \Cref{thm:connection_Shift_Counting} as an immediate corollary.

\begin{theorem}\label{thm:single_F_connection_shiftFind}
Let $A$ be a PD streaming algorithm for problem $\prac$, where $c,n>1$, 
and let $F:[0,(c+1)n]\to\set{0,1}$ be the canonical function of $A$. 
Suppose that Shift Finding with respect to this specific $F$ 
(the problem of finding an unknown shift $s^*\in[n]$
with probability at least $9/10$ given query access to $F_{s^*}$)
admits a randomized algorithm that makes
at most $q=q(F)$ (possibly adaptive) queries. 
Then the streaming algorithm $A$ must use
$\Omega(\tfrac{\log n}{\log q})$ bits of space.
\end{theorem}

\begin{proof} 
Define algorithm $A'$ to be an amplification of $A$
to success probability $1-1/(10q)$,
by running $O(\log q)$ independent repetitions and reporting their majority. 
Assume there exists an algorithm $Q$ that for every $s^*\in[n]$,
makes at most $q=q(F)$ queries to $F_{s^*}$ (possibly adaptive)
and outputs $s^*$ with probability at least $9/10$.

Consider an instance of problem $\prmsg$ with alphabet $\Sigma = [0,n]$,
and consider the following protocol for it.
Alice starts an execution of the streaming algorithm $A'$ using the shared randomness, 
then takes her input $s^*\in \Sigma$ and makes $s^*$ stream insertions to algorithm $A'$, 
and finally sends the state (memory contents) of $A'$ to Bob.

Bob continues the execution of the streaming algorithm $A'$ (using the shared randomness),
and uses it to provide query access to $F_{s^*}$, as follows.
In order to query $F_{s^*}$ at any index $x$,
Bob makes a fresh copy $A_0$ of the streaming algorithm $A'$,
insert $x$ stream items to algorithm $A_0$ and then reads its output.
With probability at least $1-1/(10q)$, 
the answer that Bob gets is indeed $F_{s^*}(x)$ (because the number of items inserted to this instance of the algorithm is $x+s^*$).
Bob uses this query access and his knowledge of $F$
to simulate algorithm $Q$ 
(with the goal of recovering $s^*$).

Consider Bob's simulation of algorithm $Q$. 
If $Q$ was executed with true query access to $F_{s^*}$,
then it would have had success probability $9/10$,
and would have made a sequence of queries $X_Q$ to $F_{s^*}$. 
This sequence $X_Q$ depends only on $F_{s^*}$ and the coin tosses of algorithm $Q$.
In particular, revealing $X_Q$ (i.e., conditioned on $X_Q$) 
does not affect the coins of the streaming algorithm $A'$,
and it still succeeds with probability at least $1-1/(10q)$. 
We can thus apply a union bound to conclude that
algorithm $A'$ succeeds on all queries $x\in X_Q$
(i.e., outputs the corresponding $F_{s^*}(x)$)
with probability at least $1-q\cdot\tfrac{1}{10q}=9/10$. 
Hence, when Bob simulates algorithm $Q$ using the streaming algorithm $A'$,
with probability $9/10$ (over the coins of $A'$)
the execution is identical to running algorithm $Q$ with true access to $F_{s^*}$,
which itself succeeds with probability $9/10$. 
By a union bound, with probability $8/10$ both algorithm $Q$ and the streaming algorithm $A'$ succeed, in which case Bob recovers $s^*$,
and therefore this communication protocol solves problem $\prmsg$ with alphabet $\Sigma = [0,n]$.

By Lemma~\ref{lem:comm_problem_message}, the message Alice sends must contain $\Omega(\log n)$ bits, and thus the streaming algorithm $A'$ must use $\Omega(\log n)$ bits of space.
Recall that algorithm $A'$ consists of $O(\log q)$ copies of the streaming algorithm $A$ and thus algorithm $A$ must use $\Omega(\tfrac{\log n}{\log q})$ bits of space.
\end{proof}

%%% Local Variables:
%%% mode: latex
%%% TeX-master: "main"
%%% End:

%% file: LBforPD_ApproximateCounting.tex
\section{Lower Bound for PD Approximate Counting}\label{sec:LB_PD_counting}

In this section, we prove Theorem~\ref{thm:LB_PD_counting}, i.e., for every $c,n>1$, we prove that every PD streaming algorithm for the approximate counting problem $\prac$ must use $\Omega_c(\sqrt{\tfrac{\log n}{\loglog n}})$ bits of space.

Let $F$ be the canonical function of a PD streaming algorithm for problem $\prac$.
Our analysis is split into two cases depending on $F$,
which informally correspond to whether a fixed pattern (like ``01'')
appears in the string $F$ at most $t$ times or not. 
These cases are analyzed using  \Cref{thm:generalized_connection_Shift_Counting,thm:single_F_connection_shiftFind}.
The overall bound will be derived by optimizing the threshold $t$
between the two cases to roughly $t=n/2^{\sqrt{\log n}}$.

\subsection{Scenario One}
In this scenario, there is a specific pattern in $F$
that appears at most $t$ times, where $t=t_c(n)$ will be set at the end of our proof.
We first consider the pattern "01" in $F$,
which corresponds to $x\in [0,(c+1)n-1]$ such that $F(x)=0$ and $F(x+1)=1$, and later generalize this pattern to a broader family.

\begin{lemma}\label{cl:hardness_few_01}
If the pattern "01" appears at most $t$ times in $F$, 
then every PD streaming algorithm for problem $\prac$ whose canonical function is $F$
must use $\Omega(\tfrac{\log (n/t)}{\loglog (cn)})$ bits of space.
\end{lemma}

\begin{proof}%[Proof of \Cref{cl:hardness_few_01}]
The proof is by a reduction from problem MESSAGE, similarly to the proof of \Cref{thm:single_F_connection_shiftFind}.
Perhaps the most delicate part is the definition of an alphabet $\Sigma$ for the MESSAGE problem $\prmsg$, and it proceeds as follows.

Given $s\in[n]$, consider the following execution of Binary Search (B.S.)
on the function $F_s$. 
Initialize $l=0$ and $r=cn+1$,
and at every iteration query $F_{s}(\floor{\tfrac{l+r}{2}})$;
if $F_{s}(\floor{\tfrac{l+r}{2}})=0$,
then $l\gets\floor{\tfrac{l+r}{2}}$, otherwise $r\gets\floor{\tfrac{l+r}{2}}$.
These iterations maintain the invariant that $F_s(l)=0$ and $F_s(r)=1$,
and after at most $\log (cn)$ iterations arrive at $r=l+1$ with the pattern "01".
Define a mapping $M:[n]\to [cn]$ such that $M(s)$ is the location where the binary search finds a "01" in $F_s$, i.e., the final index $l$; thus $F(s+M(s))=0$ and $F(s+M(s)+1)=1$.

In order to define an alphabet $\Sigma$,
consider a partitioning of $[n]$ to buckets, defined
such that items $s,s'$ are from the same bucket $B$ if and only if they are mapped to the same value $M(s)=M(s')$.
For every bucket $B$ and every $s,s'\in B$,
we know from above that $F(s'+M(s))=0$ and $F(s'+M(s)+1)=1$,
so there are at most $t$ possibilities for $s'$ (one of which is $s'=s$), and thus the size of the bucket $|B|\leq t$.
Define $\Sigma\subset [n]$ by taking one representative from each bucket.
Thus, every $s_1\neq s_2\in \Sigma$ satisfy $M(s_1)\neq M(s_2)$ and
 $|\Sigma|\geq n/t$.

Let $A$ be a streaming algorithm whose canonical function is $F$ and let algorithm $A'$ be an amplification of algorithm $A$ that succeeds with probability $1-1/(10\log (cn))$ (by making $O(\loglog (cn))$ repetitions and taking the majority).
Consider an instance of the MESSAGE problem $\prmsg$,
and proceed similarly to the proof of \Cref{thm:single_F_connection_shiftFind}.
We provide a self-contained analysis for completeness.
Alice and Bob perform the following protocol.
Alice starts an execution of algorithm $A'$ using the shared randomness.
For input $s^*\in \Sigma$, she inserts $s^*$ stream items to algorithm $A'$ and sends the state (memory contents) of this algorithm $A'$ to Bob.
In order to get query access to $F_{s^*}$ at index $x$, 
Bob makes a fresh copy $A_0$ of algorithm $A'$, continues the algorithm's execution (using the shared randomness),
inserts $x$ stream items to algorithm $A_0$ and finally reads its output.
Bob uses this query access to simulate the B.S. algorithm on $F_{s^*}$ (with the goal of recovering $M(s^*)$).
He then infers which bucket corresponds to his result, and outputs the representative of that bucket (which is $s^*$ if he recovers $M(s^*)$).

If the B.S. algorithm were executed with true query access to $F_{s^*}$, then it would have output $M(s^*)$ and would have made a sequence of queries $X_{BS}$ to $F_{s^*}$.
This sequence depends only on $F_{s^*}$, and in particular
independent of the random coins of algorithm $A'$.
Thus by a union bound, algorithm $A'$ succeeds on all queries $x\in X_{BS}$ (i.e. outputs the corresponding $F_{s^*}(x)$) with probability at least $1-\log (cn)\cdot 1/(10\log (cn))=9/10$.
Hence, when Bob simulates the B.S. algorithm using the streaming algorithm $A'$, then with probability $9/10$ the execution is identical to running the B.S. algorithm with true query access to $F_{s^*}$.
Thus with this probability $9/10$, Bob recovers $M(s^*)$, and hence outputs $s^*$, which concludes the correctness analysis of the communication protocol.

     By Lemma~\ref{lem:comm_problem_message}, the message Alice sends must contain $\Omega(\log|\Sigma|)\geq \Omega(\log (n/t))$ bits, and thus algorithm $A'$ must use $\Omega(\log (n/t))$ bits of space.
    Recall that algorithm $A'$ is made of $O(\loglog (cn))$ copies of algorithm $A$ and thus algorithm $A$ must use $\Omega(\tfrac{\log (n/t)}{\loglog (cn)})$ bits of space.
\end{proof}

\begin{remark}
This proof can be easily generalized to prove \Cref{thm:generalized_connection_Shift_Counting}.
The first extension is
by replacing the B.S. algorithm and the corresponding buckets with any deterministic algorithm $Q$ that returns a subset containing $s^*$.
In order to generalize $Q$ to any PD algorithm $Y$,
consider the canonical function of $Y$ instead of the mapping $M$, and apply the same proof.
It holds because the crucial property of the B.S. algorithm was the existence of the mapping $M$.
Then by an additional union bound, both algorithms $Q$ and $A'$ succeed with probability $8/10$ (as in the proof of \Cref{thm:single_F_connection_shiftFind}).
\end{remark}

We now generalize \Cref{cl:hardness_few_01} to a larger family of patterns in $F$, where each pattern
is characterized by a parameter $k\in [n]$, and appears at index $x\in [0,(c+1)n-k]$ such that $F(x)=0$ and $F(x+k)=1$. 
These patterns are allowed to overlap with each other (for different values of $k$).
Denote such a pattern by "$0?^{k-1}1$", 
where each question mark can represent either $0$ or $1$,
and the number of question marks is $k-1<n$.
A copy of this pattern can be found in $O(\log \tfrac{n}{k})$ queries to $F_{s^*}$ by a binary search on the grid $(0,k,...,\ceil{\tfrac{cn}{k}}k)$, since $F_{s^*}(0)=0$ and $F_{s^*}(\ceil{\tfrac{cn}{k}}k)=1$.
Hence, if there exists $k$ for which this pattern appears at most $t$ times in $F$, then the communication protocol above can be adjusted to imply that algorithm $A$ must use at least $\Omega(\tfrac{\log (n/t)}{\loglog (cn/k)})\geq \Omega(\tfrac{\log (n/t)}{\loglog (cn)})$ bits of space.
The only change in the proof is in the number of queries that Bob makes,
which affects the number of repetitions in algorithm $A'$,
and thus only affects the $\loglog$ term.

\begin{corollary}\label{cor:scenario_1}
    If for some $k\leq n$ the pattern "$0?^{k-1}1$" appears at most $t$ times in $F$, then every PD streaming algorithm for problem $\prac$ whose canonical function is $F$, must use $\Omega(\tfrac{\log (n/t)}{\loglog (cn)})$ bits of space.

\end{corollary}

\subsection{Scenario Two}

In this scenario, for every $k\leq n$ the pattern "$0?^{k-1}1$" appears at least $t$ times in $F$.

\begin{lemma}\label{lem:scenario_2}
    If for all $k\in [n]$, the pattern "$0?^{k-1}1$"  appear at least $t$ times in $F$, then every PD streaming algorithm for problem $\prac$ whose canonical function is $F$, must use $\Omega(\tfrac{\log n}{\log(cn/t)+\loglog n})$ bits of space.
\end{lemma}

\begin{proof}
In this case, there is an algorithm for the Shift Finding problem $\prsf$ using $q=O(\tfrac{cn\log n}{t})$ queries to $F_{s^*}$, as follows.

\begin{enumerate}
\item let $S=[0,n]$
\item repeat the following $\tfrac{10cn\log n}{t}$ times:
  \begin{enumerate}
  \item pick $r\in[cn]$ uniformly at random and query $F_{s^*}(r)$
  \item let $S\gets \{s\in S:\ F(s+r)= F_{s^*}(r) \}$
  \end{enumerate}
\item if $|S|=1$, return $s\in S$; else return FAIL
\end{enumerate}

The final set $S$ clearly contains the shift $s^*$.
It remains to show that all $s\neq s^*$ are removed from the set $S$ with high probability.

Fix $s\in [n], s\neq s^*$. 
There are $t$ values for $r\in[cn]$ for which $F(s^*+r)\neq F(s+r)$, as follows.
Assume without loss of generality that $s^*<s$ and denote $k = s-s^*\in [n]$.
Let $l$ be a location that corresponds to the pattern "$0?^{k-1}1$" in $F$, i.e. $F(l)=0$ and $F(l+k)=1$.
If $l\in [s^*+1,s^*+cn]$, then there is $r\in[cn]$ such that $s^*+r=l$, for which $F(s^*+r)=0\neq F(l+k) = F(s+r)$.
There are at least $t$ locations for this pattern (i.e. possible values for $l$), thus it remains to show that indeed $l\in [s^*+1,s^*+cn]$.
It must be that $l+k> n$ since $F(x)=0$ for all $x\leq n$, and similarly $l\leq cn$ since $F(x)=1$ for all $x> cn$.
Hence $l\in [n-k+1,cn]\subset [s^*+1,s^*+cn]$,
and thus there are $t$ values for $r\in[cn]$ for which $F(s^*+r)\neq F(s+r)$ (each value for $r$ corresponds to a possible value for $l$).

Thus, in each repetition, $s$ is removed from the set $S$ with probability at least $\tfrac{t}{cn}$.
The probability $s$ is not removed after $\tfrac{10cn\log n}{t}$ repetitions
is $(1-\tfrac{t}{cn})^{(10cn\log n)/t}<\tfrac{1}{n^2}$.
By a union bound, all $s\neq s^*$ are removed with probability $1-\tfrac{1}{n}$,
which concludes the correctness analysis of the algorithm for problem $\prsf$.

By \Cref{thm:single_F_connection_shiftFind}, every PD streaming algorithm for the approximate counting problem $\prac$ with a canonical function $F$ 
must use $\Omega(\tfrac{\log n}{\log((cn\log n)/t)})$ bits of space.

\end{proof}

\subsection{Concluding the Proof of Theorem~\ref{thm:LB_PD_counting}}

Concluding the two scenarios, set $t=n/2^{\sqrt{\log n \cdot \log\log (cn)}}$ and get by Corollary~\ref{cor:scenario_1} and Lemma~\ref{lem:scenario_2} that every PD streaming algorithm for the approximate counting problem $\prac$ must use
\[
\Omega(\min\{\tfrac{\log (n/t)}{\loglog (cn)}, \tfrac{\log n}{\log ((cn/t)\log n)}\})
= \Omega(\tfrac{\log n}{\sqrt{\log n\loglog (cn)} + \log c})
\]
bits of space, which boils down to $\Omega(\sqrt{\tfrac{\log n}{\loglog n}})$ for $c<2^{\sqrt{\log n \loglog n}}$.

%%% Local Variables:
%%% mode: latex
%%% TeX-master: "main"
%%% End:

%% file: Shift_Finding_Alg.tex
\section{Shift Finding Algorithm}\label{sec:shift_finding_alg}

One can hope to prove tighter lower bounds for PD streaming algorithms for the approximate counting problem $\prac$, and a possible approach is by solving the Shift Finding problem $\prsf$ using $\polylog n$ queries.
Recall that in problem $\prsf$, the input is a string $P\in\set{0,1}^{(c-1)n}$, 
which can be represented by a string $F$ which is a concatenation of $n$ zeros, $P$ and then $n$ ones; and query access to a shifted version of $F$ with shift $s^*$, denoted $F_{s^*}$.
As stated in Theorem~\ref{thm:shift_finding_alg},
we show a deterministic algorithm for problem $\prsf$ using $O(\sqrt{cn})$ queries (Algorithm~\ref{alg:sqrt_n_shift_finding}), and we leave open the question whether it is the right bound.
 The proof relies on an efficient verification algorithm that for input $s$, uses $2$ queries and returns 'yes' if and only if $s=s^*$, as stated in Lemma~\ref{lem:shift_verifier_witness_2queries} and described next.
\begin{proof}[Proof of Lemma~\ref{lem:shift_verifier_witness_2queries}]
    Denote by $l\in [n+1,cn+1]$ the smallest number such that $F(l)=1$, and by $r\in[n,cn]$ the largest number such that $F(r)=0$.
    For input $s\in[0,n]$, the verification algorithm
    returns 'no' if $F_{s^*}(l-s)=0$ or $F_{s^*}(r-s)=1$, and
    otherwise returns 'yes'.

    If $s=s^*$, then $F_{s^*}(x-s)=F(x)$ and the verification algorithm outputs 'yes'.
    If $s> s^*$, then $s^*-s+l<l$ and thus $F_{s^*}(l-s) = F(s^*-s+l)=0$ 
    and the verification algorithm outputs 'no'.
    Similarly, if $s< s^*$ then $F_{s^*}(r-s)=1$ and the verification algorithm outputs 'no'.
\end{proof}

\begin{remark}
There is a randomized algorithm for problem $\prsf$ using $\tilde{O}_c(\sqrt{n})$ queries that is similar to
the proof of Theorem~\ref{thm:LB_PD_counting} in \Cref{sec:LB_PD_counting}.
It proceeds by considering those two scenarios.
In scenario one, instead of constructing the set $\Sigma$, query witnesses for all the $t$ possible shifts using $2t$ queries and hence recover the unknown shift $s^*$.
In scenario two, the proof of Theorem~\ref{thm:LB_PD_counting} shows how to find the unknown shift $s^*$ in $O(\tfrac{cn}{t}\log n)$ queries with high probability.
Hence, by setting $t=\sqrt{cn\log n}$, this algorithm finds the unknown shift in $O(\max \{t+\log (cn), \tfrac{cn}{t}\log n\}) \leq O(\sqrt{cn\log n})$ queries with high probability.
\end{remark}

Next is a slight improvement, 
a deterministic algorithm in $O(\sqrt{cn})$ queries, proving Theorem~\ref{thm:shift_finding_alg}.

\begin{algorithm}
    \caption{Deterministic Shift Finding in $O(\sqrt{cn})$ queries}\label{alg:sqrt_n_shift_finding}

    \begin{algorithmic}[1]
        \Require $n,c,F$ and query access to $F_{s^*}$
        \Ensure $s^*$

        \State $Q \gets (F_{s^*}(0),F_{s^*}(\sqrt{cn}),F_{s^*}(2\sqrt{cn}),...,F_{s^*}(cn))$
        \State let $S\gets\set{s\in[0,n]:\forall i\in[0,\sqrt{cn}], F_s(i\sqrt{cn})=Q(i)}$  \Comment{i.e. the set of all shifts that could have produced $Q$}
        \For{$s\in S$}
            \State check the witness of $s$
            \If{$s=s^*$} return $s$
            \EndIf
        \EndFor
    \end{algorithmic}
\end{algorithm}

\begin{lemma}
    The set $S$ in Algorithm~\ref{alg:sqrt_n_shift_finding} is of size $O(\sqrt{cn})$.
\end{lemma}
\begin{proof}
    Assume by contradiction that $|S|\geq \sqrt{cn}+1$.
    Hence by the pigeonhole principle, there exists $s_1<s_2\in S$ such that $s_1=s_2 \mod \sqrt{cn}$.
    Hence for all $i\in[0,\sqrt{cn}-\tfrac{s_2-s_1}{\sqrt{cn}}]$,
    \begin{align*}
        Q(i) = F_{s_2}(i\sqrt{cn})
            = F_{s_1}(s_2-s_1 + i\sqrt{cn})
            = Q(\tfrac{s_2-s_1}{\sqrt{cn}} + i),
    \end{align*}
    where the first and last transitions hold since $s_1,s_2\in S$ and $\tfrac{s_2-s_1}{\sqrt{cn}}$ is an integer number, and the second transition is by definition.
    Thus $Q$ has a period of length $\tfrac{s_2 - s_1}{\sqrt{cn}}\leq \lfloor\tfrac{s_2}{\sqrt{cn}}\rfloor$.
    However, for $i\in [\sqrt{cn}-\lfloor\tfrac{s_2}{\sqrt{cn}}\rfloor +1 , \sqrt{cn}]$ the values that $Q$ get are $Q(i) = F_{s_2}(i\sqrt{cn})=1$ since $s_2 + i\sqrt{cn}\geq cn$;
    % the last $\lfloor\tfrac{s_2}{\sqrt{n}}\rfloor$ entries in $Q$ are equal $1$,
    thus all entries in $Q$ are equal $1$, which contradicts the fact that $Q(0)=0$, and thus completes the proof.
\end{proof}

Algorithm~\ref{alg:sqrt_n_shift_finding} returns the shift $s^*$ since $s^*\in S$ and by the correctness of the verifier in Lemma~\ref{lem:shift_verifier_witness_2queries}.
The number of queries Algorithm~\ref{alg:sqrt_n_shift_finding} makes is $O(|S|+|Q|)=O(\sqrt{cn})$, which proves Theorem~\ref{thm:shift_finding_alg}.

%%% Local Variables:
%%% mode: latex
%%% TeX-master: "main"
%%% End: